\documentclass{amsart}
\usepackage{amsmath}
\usepackage{amsfonts}
\usepackage{amssymb}
\usepackage{graphicx}

\newtheorem{theorem}{Theorem}
\theoremstyle{plain}

\newtheorem{example}{Example}

\newtheorem{remark}{Remark}

\numberwithin{equation}{section}

\begin{document}
\title[SAFIP algorithm]{SAFIP: A Streaming Algorithm for Inverse Problems}
\author{Maeva Biret$^{(1)}$}
\author{Michel Broniatowski$^{(2,\ast )}$}
\curraddr{$^{(1)}$Safran Aircraft Engine, Moissy-Cramayel}
\curraddr{$^{(2)}$Universit\'{e} Pierre et Marie Curie, Paris}
\email{$^{(\ast )}$ Corresponding author: michel.broniatowski@upmc.fr}
\date{May 26th, 2016}

\begin{abstract}
This paper presents a new algorithm which aims at the resolution of inverse problems of the form $f(x)=0$, for $x\in\mathbb{R}^d$ and $f$ an arbitrary function with mild regularity condition. The set of solutions $S$ may be infinite. This algorithm produces a good coverage of $S$, with a limited number of evaluations of the function $f$. It is therefore appropriate for complex problems where those evaluations are costly. Various examples are presented, with $d$ varying from 2 to 10. Proofs of convergence and of coverage of S are presented. 
\end{abstract}

\maketitle

\section{Introduction}
\label{SECTION INTRO} 
\subsection{The scope of this paper}
Assume that we are given a bounded and closed domain $D\subset\mathbb{R}^d$, and a continuous real-valued function $f$ defined on $D$.\\
The aim of this paper is to present an algorithm for the solution of the problem
\begin{equation}\label{EQ pb}
S=\lbrace x\in D : f(x)=0\rbrace,
\end{equation}
assuming $S\not =\emptyset$.\\
Such problems have been extensively handled over the years; see \cite{Nakamura}. The difficulty which we are confronted to lies in three main points : 
\begin{enumerate}
\item the set $S$ may contain many points, even be infinite,
\item the function $f$ might be quite costly for example when defined by a simulation device,
\item the function $f$ may be quite irregular; we will assume mild regularity in the neighborhood of any point in $S$, only.
\end{enumerate}
We also provide a two-fold proof for the convergence of this algorithm, namely we first prove that any resulting sequence of points in $D$ converges to some point in $S$, and secondly that any point $x$ in $S$ is reached asymptotically by some "good" sequence, which is a sequence starting in a suitable neighborhood of $x$. As usually done in random search techniques, the starting points will be defined through random sampling in $D$.

\subsection{Bibliographic outlook}
Most approaches to Problem (\ref{EQ pb}) extensively use analytic properties of the function $f$; dichotomy, false position, Newton, conjugate gradient, etc (see \cite{Suli}) handle so called well-posed problems, when the equation $f(x)=0$, for $x\in\mathbb{R}$ and $f$ a real-valued function, has a unique solution. The case where $f$ is defined as a mapping from $\mathbb{R}^d$ to $\mathbb{R}^k$ with $d\leq k$ is treated by singular value decomposition (see \cite{Golub}), which also solves well-posed problems.\\
The ill-posed problems which we consider, namely the case where Problem (\ref{EQ pb}) has multiple solutions, is usually handled through regularization techniques, which aim at transposing (\ref{EQ pb}) into a well-posed problem. This procedure produces a partial solution to (\ref{EQ pb}) under appropriate knowledge on the function $f$ (see \cite{Tikhonov}). All these techniques are out of the concern of the present work, where all solutions of $f$ are looked for, with minimal assumption on $f$. We briefly present four methods, which constitute the environment of our proposal.

\paragraph{\textit{Local multi-start optimization, a deterministic approach}}
Looking for the value of $x$ such that $f(x)=0$, consider the function $\vert f\vert$; minimizing $\vert f\vert$ indeed produces the set $S$.\\
First we choose a local optimization technique (Newton-Raphson for example). Then consider a design, which is a grid of initial points for the local optimization. From any of those, the sequence of iterations of the local optimization algorithm may produce a limit solution in $S$. Obviously stationary points not in $S$ may be produced. The initial design is of utmost importance and the method may be unstable in this respect. Furthermore the method may be very costly due to the numerous evaluations of $f$. A general reference for those methods is \cite{Gyorgy}.

\paragraph{\textit{A grid search, deterministic approach}}
This method produces a sequence of grids in $D$. Given an initial regular grid, the function $f$ is evaluated on each of its points. Points where $f$ is close to 0 are selected and the grid is updated and refined in the neighborhood of those points. This method has been proposed by \cite{Miller2}. A serious drawback lies in its cost, when the dimension of $D$ corresponds to real life cases. Furthermore, the stopping rule of such algorithms does not guarantee a uniform approximation of $S$.

\paragraph{\textit{A Monte Carlo Markov Chain technique}}
We assume that the function $f$ is written as $f(x)=g(x)+\epsilon$. $f$ is then a model for the real function $g$ with an error $\epsilon$ due to modelling. For example, $g$ is a physical model and $f$ a computer-based formula for $g$. We estimate $S=\lbrace x: g(x)=0\rbrace$. We choose a prior distribution $\Pi_0(x)$ on $\mathcal{X}$ and a parametric form for the distribution of $\epsilon$, $p(\epsilon\vert x)$, for fixed $x$. By Bayes formula, the a posteriori distribution of $x$ given $\epsilon$ is given by
\begin{equation}\label{EQ posterior}
\frac{p(\epsilon\vert x)\Pi_0(x)}{\int p(\epsilon\vert x)\Pi_0(x)dx}.
\end{equation}
The maximum probability principle provides stochastic solutions of $g(x)=0$ as the maximum of (\ref{EQ posterior}) upon $x$, given the prior $\Pi_0$.\\
In turn it can be proved that, whenever $\Pi_0(x)=\mathcal{N}(x_0, \sigma_0^2)$ the Gaussian distribution with mean $x_0$ and variance $\sigma_0^2$, for some $x_0\in D$ and $\sigma_0^2>0$, solutions $x^*$ of (\ref{EQ posterior}) can be written as
\begin{equation}\label{EQ MCMC}
x^*:=argmin_{x\in D}\vert\vert y-g(x)\vert\vert+\frac{\sigma_\epsilon^2}{\sigma_0^2}\vert\vert x-x_0\vert\vert^2,
\end{equation}
when $\epsilon$ is assumed to follow $\mathcal{N}(0, \sigma_\epsilon^2)$.\\
In order to find the $x^*$ solution of (\ref{EQ MCMC}), MCMC routines are used. This method is described in \cite{Gelfand}.

\paragraph{\textit{The MRM (Monotonous Reliability Method)}}
Assume that $f: \mathbb{R}^d\rightarrow\mathbb{R}$ is a globally monotone, i. e. is monotone in each of its variables. Assume also that the set $S$ of solutions of the equation $f(x)=0$ is a continuous and simply (or one) connected set.\\
Assume for example that $f$ is increasing on each of its variables. At each step, choose one point $x$ in the unexplored subset of $D$. When $f(x)>0$ then all points $y>x$ (meaning $y_i>x_i$ for all $1\leq i\leq d$) are discarded from the unexplored region. \\
In the same way, when $f(x)<0$, discard all the regions $\lbrace y: y<x\rbrace$.\\
Iteration of these steps produces an unexplored domain which shrinks to $S$.\\
Various ways of choosing $x$ in the unexplored domain define specific algorithms. See \cite{Biret}.

\section{Outlook of the SAFIP algorithm}
\label{SECTION OUTLOOK}

\subsection{Basic features and properties}\label{sec algo}
We start with the iteration of the equivalence
\begin{equation}\label{EQ d=1}
(f(x)=0)\iff\left(f(x)+\frac{x}{2k}+\frac{x}{2k}=\frac{x}{k}\right),
\end{equation}
which holds where $d=1$, for any $k\not=0$; for sake of convenience state $k>0$.\\
We proceed defining a recurrence in the RHS in (\ref{EQ d=1}), namely define a sequence $(z_i)_{i\in\mathbb{N}}$ with $z_i\in D$ and such that
\begin{equation}\label{EQ recur}
z_{i+1}=z_i+\frac{z_{i-1}-z_i}{2}+kf(z_i).
\end{equation}
Defining
\begin{equation}\label{EQ dist}
R_i=\vert z_i-z_{i-1}\vert,
\end{equation}
we obtain from (\ref{EQ recur})
\begin{equation}\label{EQ ineg dist}
R_{i+1}\leq\frac{R_i}{2}+k\vert f(z_i)\vert.
\end{equation}
When $d>1$, we may write
\[
R_i=\vert\vert z_i-z_{i-1}\vert\vert.
\]
Thus, any sequence $(z_i)$ which satisfies (\ref{EQ recur}) also satisfies (\ref{EQ ineg dist}). We define $R_0>0$ arbitrary.\\
We now propose to substitute (\ref{EQ recur}) by a random sequence $(z_i)$ which satisfies (\ref{EQ ineg dist}). Also some additional conditions on $(z_i)$ will be imposed. We will thus be able to prove the convergence of the resulting sequence $(z_i)$ to some point in $S$; reciprocally, for any $x$ in $S$, when $z_0$ is close enough to $x$, the limit point of $(z_i)$ will coincides with $x$.\\
Define $z_0$ and $z_1$ uniformly in $D$ and $R_1=\vert\vert z_1-z_0\vert\vert$.\\
For $i\geq 1$ compare $f(z_i)$ and $f(z_{i-1})$. Let $C\in\lbrack\frac{1}{2}, 1\rbrack$. If 
\begin{equation}\label{EQ decr f}
\vert f(z_i)\vert\leq C\vert f(z_{i-1})\vert,
\end{equation}
then obtain $z_{i+1}$ by
\begin{equation}\label{EQ rec}
z_{i+1}:=z_i+u_i,
\end{equation}
where $u_i$ is randomly drawn on $\mathcal{B}\left(\underline{0}, \frac{R_i}{2}+k\vert f(z_i)\vert\right)$, where $\mathcal{B}(\omega, r)$ is the ball with center $\omega$ and radius $r$.\\
If (\ref{EQ decr f}) is not fulfilled then the sequence $(z_j)_{j\in\mathbb{N}}$ stops. Draw then $z_0$ and $z_1$ again.\\
At this point we state
\begin{theorem}\label{thm1}
Any infinite sequence $(z_i)$ defined as above converges a. s. with limit in $S$.\\
\end{theorem}
We now add a number of conditions on the function $f$ which entail that any point in $S$ is reached asymptotically.\\
Let $x\in S$ and set $z_0\in\mathcal{B}(x, \epsilon_0)=\lbrace z: \vert\vert z-x\vert\vert\leq\epsilon_0\rbrace$ for some $\epsilon_0>0$. Define further
\begin{equation}\label{EQ_E0}
E_0:=B\cap\lbrace z: \vert\vert z-z_0\vert\vert>k_1\vert f(z_0)\vert\rbrace,
\end{equation}
with $0<k_1<k$ and such that $k_1\vert f(z_0)\vert<2\epsilon_0$; $B$ is the ball with center $z_0$ and radius $\frac{R_0}{2}+k\vert f(z_0)\vert$. Therefore, $E_0$ is an annulus around $z_0$.\\
Let
\begin{equation}\label{EQ A1}
A_1=int\lbrace\mathcal{B}(x, \epsilon_0)\cap B\rbrace.
\end{equation}
By its very definition, the set $A_1$ is not void.\\
Assume that $f$ satisfies the following regularity conditions
\begin{enumerate}
\item\label{cond1} For all $x\in S$, there exists some $\epsilon_0(x)>0$ such that if $z_0, z_1\in\mathcal{B}(x,\epsilon_0)$ and $\vert\vert x-z_1\vert\vert\leq\vert\vert x-z_0\vert\vert$ then
\[
\lbrace z : \vert f(z)\vert\leq \vert f(z_1)\vert\rbrace\varsubsetneq\lbrace z : \vert f(z)\vert\leq \vert f(z_0)\vert\rbrace.
\]
\item\label{cond3} There exists $0<m<\frac{1}{4\epsilon_0}$ such that for all $x\in S$, for all $z_0\in\mathcal{B}(x, \epsilon_0)$ for all $0<\epsilon<k/2$, for all $z\in E_0\cap A_1$,
\[
\vert f(z_0)\vert-\vert f(z)\vert\geq m\vert\vert z-z_0\vert\vert.
\]
By condition (\ref{cond1}), the LHS in this inequality is non negative.
\end{enumerate}
We then have
\begin{theorem}\label{thm2}
Let $x\in S$ and $\epsilon_0>0$ such that (\ref{cond1}) and (\ref{cond3}) hold. When $z_0\in\mathcal{B}(x, \epsilon_0)$, the sequence $(z_i)$ is infinite and satisfies Theorem \ref{thm1}. Furthermore $\lim z_i=x$ a. s.
\end{theorem}
In order to cover all $S$ by the limiting points of such sequences we also propose to add a step where we randomly select $p$ points uniformly in $D$. These points are initial points of new sequences; this allows to obtain a good covering of $S$ by the limits of all these sequences.\\
Obviously this latest step does not substitute the entire algorithm; clearly a hudge number of such points will approximate $S$ from the start, the most inefficient Monte-Carlo random search method.\\
\newline The stopping rule is defined through the definition of an accuracy index call $tol$. Define $N$ the number of points to be reached in $S$. We may decide to stop the algorithm when $N$ sequences $(z_i)$ are such that the extremities are in $S$ up to the accuracy, denoted $tol$ in the sequel.

\subsection{Enhanced algorithm}\label{sec algo}
In order to improve the coverage of $S$, keeping the same set of points $z_0$, we propose to modify the choice of $z_{i+1}$ as given in (\ref{EQ decr f}) and (\ref{EQ rec}) as follows. From $z_0, \ldots, z_i$ we build indeed $i$ chains, each one starting from $z_j, 1\leq j\leq i$. Obviously the sequence starting at $z_i$ is as described above; the new $i-1$ ones spread and develop in all directions. Any of these chains inherit of the properties mentioned in Theorem \ref{thm1}. Also, any $x$ in $S$ is asymptotically reached by one of those sequences, as $i$ increases.\\
The sequences defined by an algorithm may be finite; indeed condition (\ref{EQ decr f}) may not hold for $(z_{i-1}, z_i)$ and therefore $z_{i+1}$ cannot be simulated. Thus no point $z_{i+1}$ will be simulated since his father would be higher than his grandfather.\\
However his grandfather $z_{i-1}$ is indeed lower than his grand-grandfather; therefore his grandfather may have offspring. This grandfather is the root of a new generation, hence a new $z_i$ which may satisfy (\ref{EQ decr f}). In the same way all ancestors of $z_{i-1}$ satisfy (\ref{EQ decr f}) and are eligible for fatherhood.\\
We call a step of the algorithm the generation of all the offspring of the eligible points in the existing population of points. Such a step is followed by the generation of $p$ uniformly distributed points in $D$ as done in the basic algorithm.\\
\newline In the sequel, we focus on the basic algorithm described in Section \ref{sec algo}.

\subsection{Reducing the computational cost tuning the parameters}
Firstly this algorithm makes use of very few parameters. Furthermore those can be tuned easily according to the complexity of the problem at hand. Indeed these parameters can be interpreted in connection with the computational burden. In some cases the function $f$ is very costly and running an algorithm for a long time, without evaluating $f$ often, may be of great advantage. Sometimes the function $f$ is easy to calculate and the need is to get a quick description of $S$. Tuning $k, C$ and $m$, together with the number of initiating points, makes use for those choices.
\newline The following examples illustrate the role of each of the parameters, all the other ones being kept fixed.\\
The number of solutions which we require in the tolerance zone around $S$ is fixed to 1000, but in the last example where the algorithm is evaluated with respect to this number.\\
Examples are presented in dimension 2. Higher dimension examples are presented in Section \ref{sec dim+}.
Red points are couples $(x_1, x_2)$ such that $f(x_1, x_2)>0$. Points with negative values of $f$ are blue. Black points are all blue or red ones whose $f$ value belongs to $[-tol, tol]$.\\
Each example is summarized by three indicators. The first one is the runtime. The second one is the efficiency coefficient (EC) which is the ratio between the total number of evaluations of $f$ and the number of solutions, which equals 1000 in all but the last example. This indicator is a measure of the number of calls to $f$ which are required in order to obtain one solution to the equation $f(x)=0$. The third indicator is of visual nature; in all those examples which are in dimension 2, the quality of the coverage of $S$ can be considered qualitatively.

\begin{remark}
The most important indicator is EC, since in all industrial applications, what really matters is the cost in evaluating $f$.
\end{remark}

\hspace{1pt}
\paragraph{\textit{The initialization step}}
Call $n$ the number of initiating points $z_0$, randomly selected on $D$. This is the initial cost of the method since the function $f$ will be evaluated $n$ times. Due to section \ref{sec algo}, $n$ should not be too large.
\begin{example}
Let $f$ be a bivariate function defined by
\[
(x_1, x_2)\mapsto f(x_1, x_2)=x_1^2+x_2^2-0.5
\]
The aim is to find $N=1000$ pairs $(x_1, x_2)$ such that $\vert f(x_1, x_2)\vert\leq tol$ where $tol$ is the accuracy. All parameters but $n$ are fixed. The tolerance is 0.01; the value of $C$ is fixed being 0.75; the value of $k$ is 1; the number $p$ of supplementary points at each step of the algorithm is 1.\\
The solutions are close to $S=\lbrace(x_1, x_2), f(x_1, x_2)=0\rbrace$, the circle with center $(0, 0)$ and radius $\sqrt{0.5}$. In Figure \ref{fonction_quad}(a), the function $f$ is intersected by the horizontal plane $z=0$. The Figure \ref{fonction_quad}(b) represents the intersection in the variables frame. The circle is then clearly visible.
\begin{figure}[!h]
\centering
\includegraphics[scale=0.4]{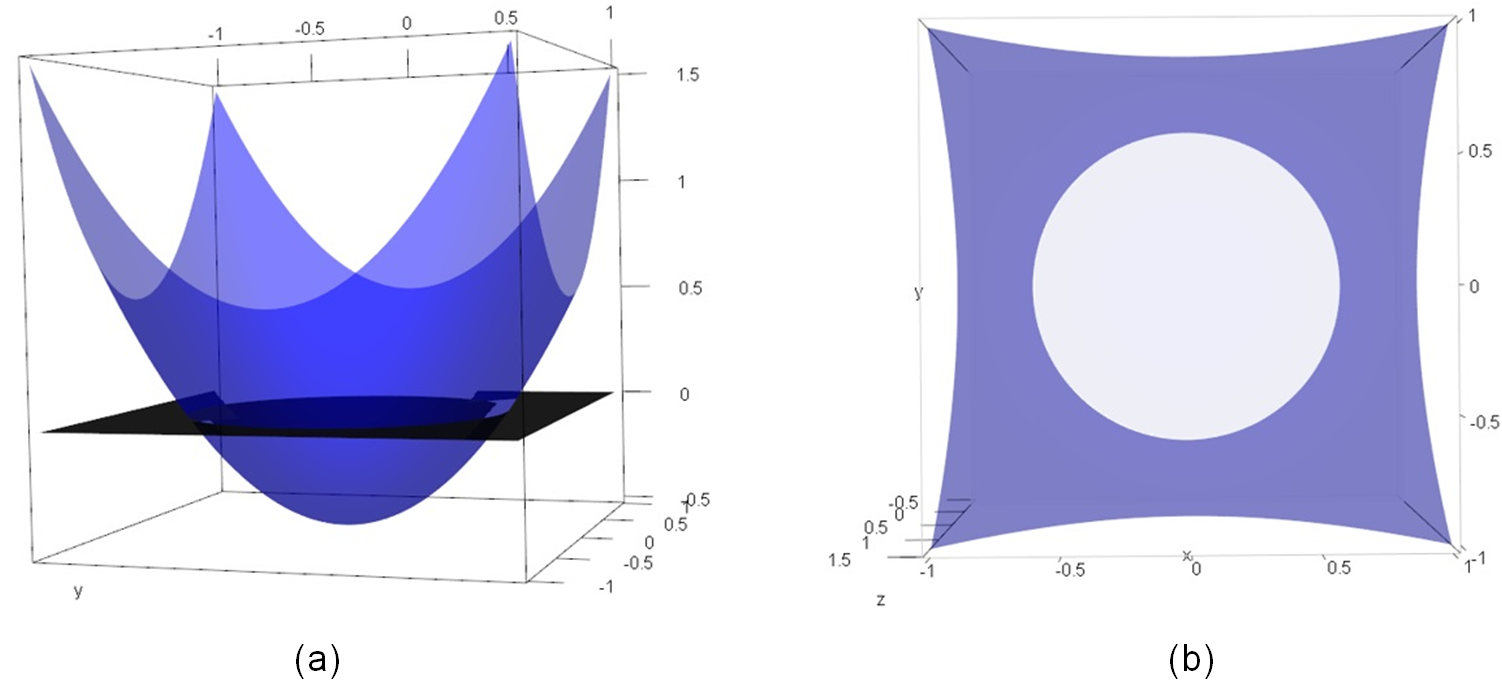}
\caption{Representations of the quadratic function}
\label{fonction_quad}
\end{figure}
In Figures \ref{ex_n_init} (a), (b), (c), we have considered respectively $n=5$, $n=100$ and $n=300$.
\begin{figure}[!h]
\centering
\includegraphics[scale=0.5]{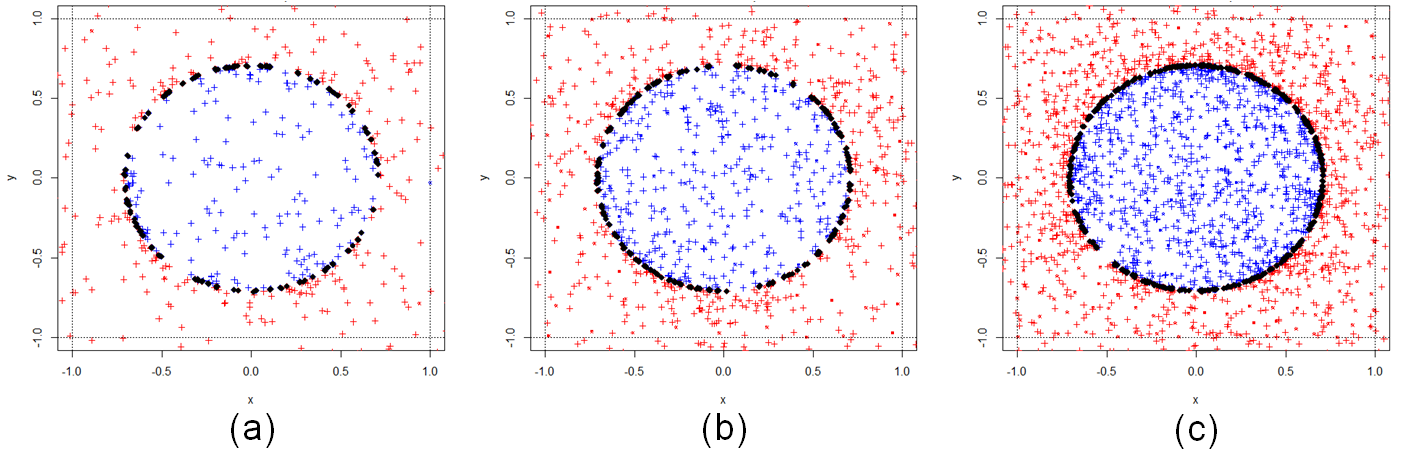}
\caption{Solving quadratic equation using SAFIP for three values of $n$}
\label{ex_n_init}
\end{figure}
Clearly the more numerous the initial points, the more the number of chains, and therefore the more numerous the points where the function $f$ is evaluated; so the algorithm is costly as $n$ increases. At the contrary, the better the coverage of $S$. Results are gathered in Table \ref{table_ex1}.
\begin{table}
\begin{center}
\begin{tabular}{|c|c|c|c|c|c|c|c|c|}
\hline
\textbf{n} & \textbf{tol} & \textbf{N} & \textbf{C} & \textbf{k} & \textbf{p} & \textbf{Time} & \textbf{EC} & \textbf{Coverage} \tabularnewline
\hline
5 & 0.01 & 1000 & 0.75 & 1 & 1 & 0.32s & 4.33 & - \tabularnewline
\hline
100 & 0.01 & 1000 & 0.75 & 1 & 1 & 0.60s & 6.32 & + \tabularnewline
\hline
300 & 0.01 & 1000 & 0.75 & 1 & 1 & 1.54s & 9.14 & ++ \tabularnewline
\hline
\end{tabular}
\end{center}
\caption{Results for Example 1 with different values of $n$}
\label{table_ex1}
\end{table}
\end{example}

\hspace{1pt}
\paragraph{\textit{The rate of convergence}}
The value of $C$ pertains to the rate of convergence of the algorithm. Assume $C$ small ($C$ close to 1/2); thus condition (\ref{EQ decr f}) is  rarely satisfied. The selected points will define chains with a fast convergence to $S$. However in order to satisfy (\ref{EQ decr f}), many simulations in the ball $B$ are required, leading to an increased runtime.
\begin{example}
Let $f$ be a bivariate function defined by
\[
(x_1, x_2)\mapsto f(x_1,x_2)=x_1^4+x_2^3-0.5
\] 
The aim is to find $N=1000$ pairs $(x_1, x_2)$ such that $\vert f(x_1, x_2)\vert\leq tol$ where $tol$ is the accuracy. All parameters but $C$ are fixed. The number of initial points is 10; the tolerance is 0.015; the value of $k$ is 1; the number $p$ of supplementary points at each step of the algorithm is 1.\\
In Figure \ref{fonction_fauteuil}(a), the function $f$ is intersected by the horizontal plane $z=0$. The Figure \ref{fonction_fauteuil}(b) represents the intersection in the variables frame.
\begin{figure}[!h]
\centering
\includegraphics[scale=0.4]{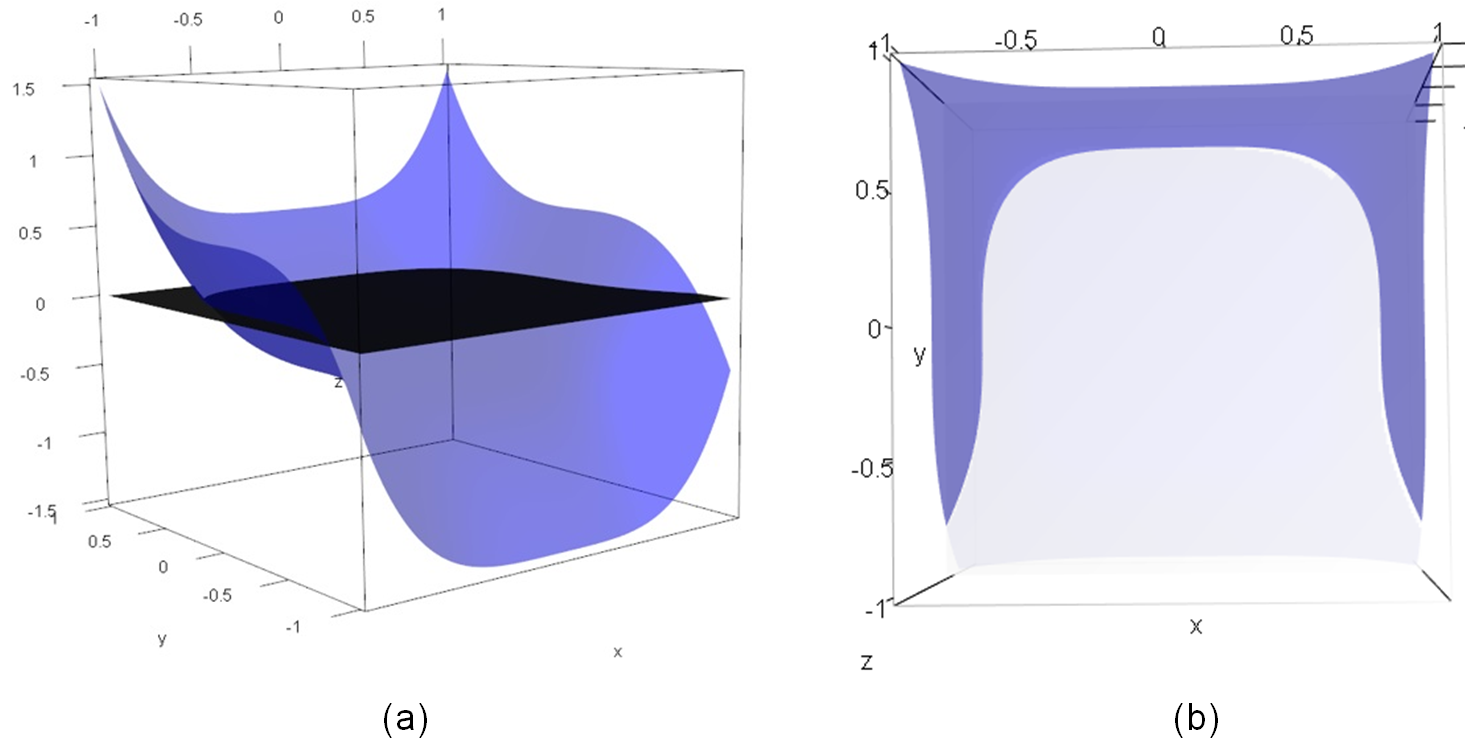}
\caption{Representations of the function with a chair shape}
\label{fonction_fauteuil}
\end{figure}
In Figures \ref{ex_C} (a), (b), (c), we have considered respectively $C=0.55$, $C=0.75$ and $C=0.95$.
\begin{figure}[!h]
\centering
\includegraphics[scale=0.5]{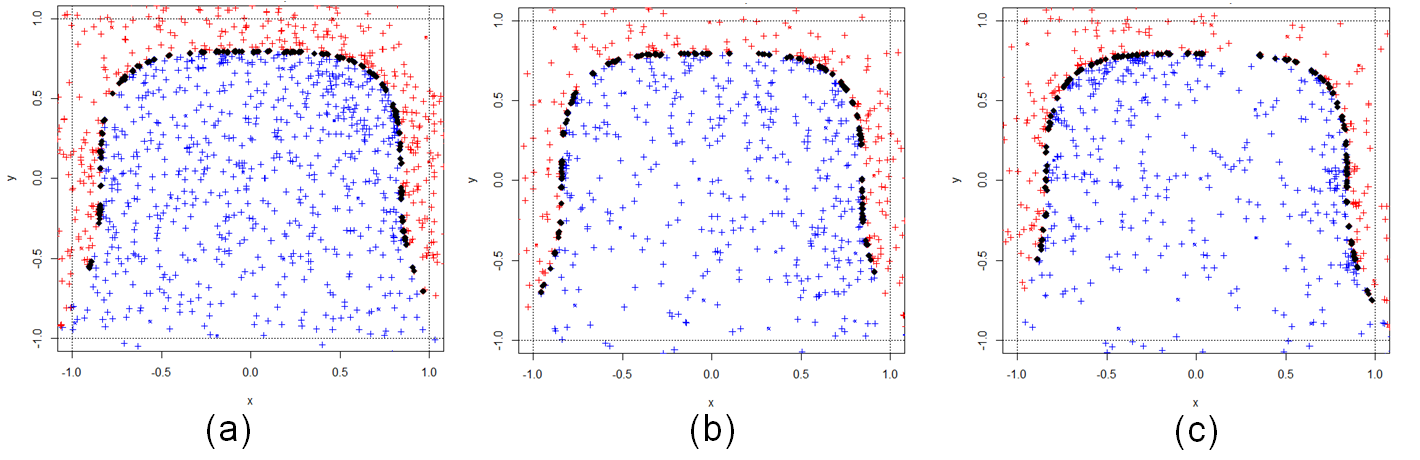}
\caption{Solving equation for the function with a chair shape using SAFIP for three values of $C$}
\label{ex_C}
\end{figure}
The greater $C$, the less the number of evaluations of $f$; furthermore the runtime decreases as $C$ increases. Results are gathered in Table \ref{table_ex2}.
\begin{table}
\begin{center}
\begin{tabular}{|c|c|c|c|c|c|c|c|c|}
\hline
\textbf{n} & \textbf{tol} & \textbf{N} & \textbf{C} & \textbf{k} & \textbf{p} & \textbf{Time} & \textbf{EC} & \textbf{Coverage} \tabularnewline
\hline
10 & 0.015 & 1000 & 0.55 & 1 & 1 & 0.62s & 8.36 & + \tabularnewline
\hline
10 & 0.015 & 1000 & 0.75 & 1 & 1 & 0.44s & 5.33 & + \tabularnewline
\hline
10 & 0.015 & 1000 & 0.95 & 1 & 1 & 0.42s & 5.05 & + \tabularnewline
\hline
\end{tabular}
\end{center}
\caption{Results for Example 2 with different values of $C$}
\label{table_ex2}
\end{table}
\end{example}

\hspace{1pt}
\paragraph{\textit{The role of $k$}}
The parameter $k$ is crucial for the simulation around $z_i$. In order to give some insight on the value of $k$, suppose that $z$ belongs to $\lbrack-1,1\rbrack^2$, and that the mean value of $\vert f(z)\vert$ is $\bar{f}=10$. The current radius of the ball $B$ is $\frac{R}{2}+k\vert f(z)\vert$, with $R$ the distance between two points in the chain. Thus $k$ should be at most of order $\frac{1}{\bar{f}}$; in this way the ball $B$ lays in $\lbrack-1, 1\rbrack^2$, roughly.\\
This appears clearly in Example 3.
\begin{example}
Let $f$ be a bivariate function defined by
\[
(x_1, x_2)\mapsto f(x_1,x_2)=(1-x_1)^2+100(x_2-x_1^2)^2-50
\] 
The aim is to find $N=1000$ pairs $(x_1, x_2)$ such that $\vert f(x_1, x_2)\vert\leq tol$ where $tol$ is the accuracy. All parameters but $k$ are fixed. The number of initial points is 10; the tolerance is 3; the value of $C$ is 0.75; the number $p$ of supplementary points at each step of the algorithm is 1.\\
In Figure \ref{fonction_rosen}(a), the function $f$ is intersected by the horizontal plane $z=0$. Figure \ref{fonction_rosen}(b) represents the intersection in the variables frame.
\begin{figure}[!h]
\centering
\includegraphics[scale=0.4]{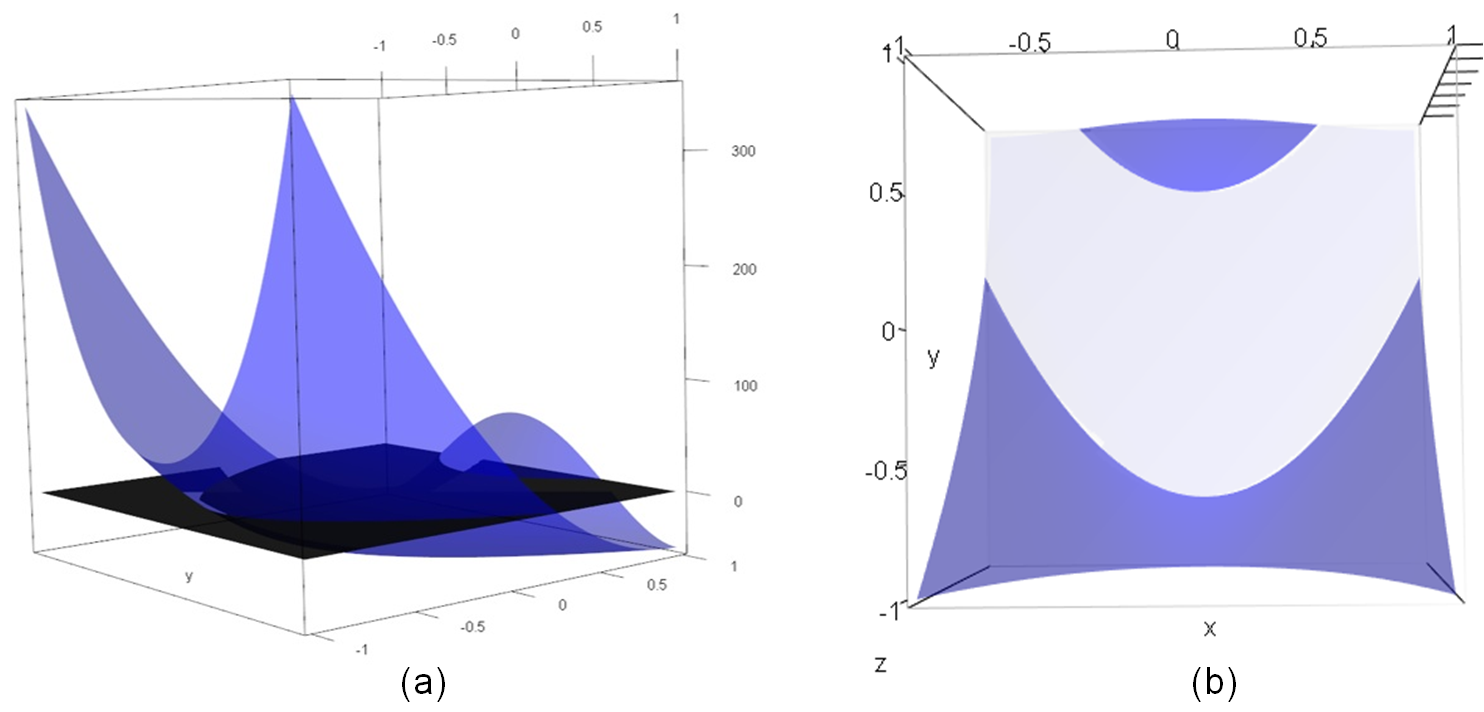}
\caption{Representations of the Rosenbrock function}
\label{fonction_rosen}
\end{figure}
The mean value of $f$ is 200 and its variations belong to $\lbrack-50,350\rbrack$. In Figures \ref{ex_k} (a), (b), (c), we have considered respectively $k=1/200$, $k=10/200$ and $k=50/200$.
\begin{figure}[!h]
\centering
\includegraphics[scale=0.5]{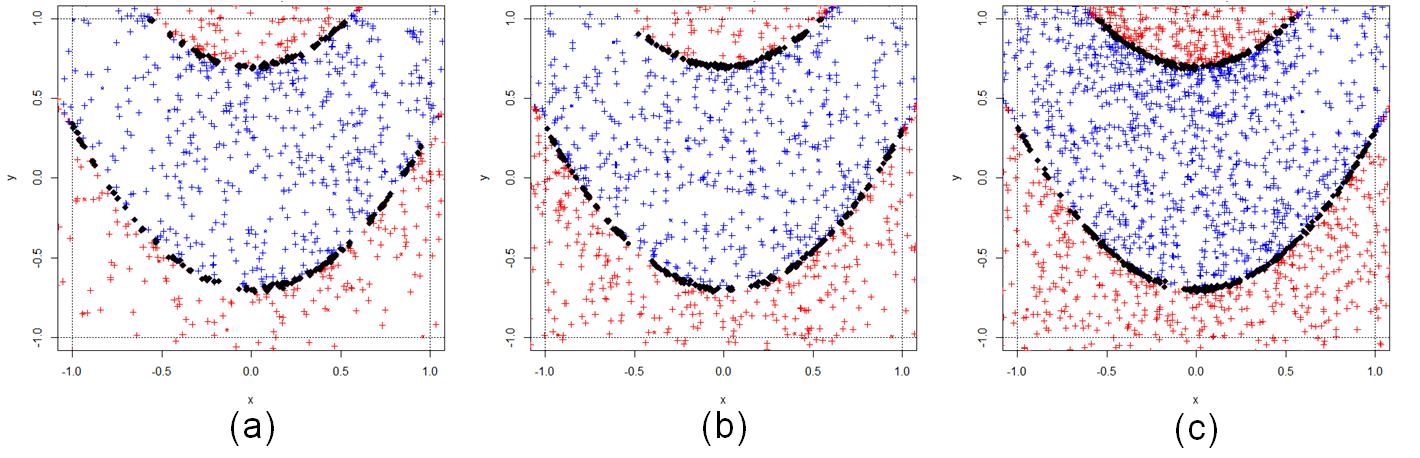}
\caption{Solving equation for the Rosenbrock function using SAFIP for three values of $k$}
\label{ex_k}
\end{figure}
As $k$ increases, the runtime also increases as does the number of evaluations of $f$ in order to obtain one solution, and also the coverage of $S$ improves. When $f$ is costly, $k$ should be chosen small. Results are gathered in Table \ref{table_ex3}.
\begin{table}
\begin{center}
\begin{tabular}{|c|c|c|c|c|c|c|c|c|}
\hline
\textbf{n} & \textbf{tol} & \textbf{N} & \textbf{C} & \textbf{k} & \textbf{p} & \textbf{Time} & \textbf{EC} & \textbf{Coverage} \tabularnewline
\hline
10 & 3 & 1000 & 0.75 & 0.005 & 1 & 0.76s & 10.69 & + \tabularnewline
\hline
10 & 3 & 1000 & 0.75 & 0.05 & 1 & 2.76s & 18.71 & + \tabularnewline
\hline
10 & 3 & 1000 & 0.75 & 0.25 & 1 & 4.16s & 48.49 & ++ \tabularnewline
\hline
\end{tabular}
\end{center}
\caption{Results for Example 3 with different values of $k$}
\label{table_ex3}
\end{table}
\end{example}

\hspace{1pt}
\paragraph{\textit{The role of $p$}}
The number of intermediate points is important since it allows to explore new points of $D$ in quest for $S$. This number should be chosen small with respect to the number $n$ of initializing points. The following example shows that very small values of $p$ may be good choices.
\begin{example}
Let $f$ be a bivariate function defined by
\[
(x_1, x_2)\mapsto f(x_1,x_2)=(x_1-0.5)^2+3x_1x_2-x_2^3-2.25
\] 
The aim is to find $N=1000$ pairs $(x_1, x_2)$ such that $\vert f(x_1, x_2)\vert\leq tol$ where $tol$ is the accuracy. All parameters but $p$ are fixed. The number of initial points is 10; the tolerance is 0.04; the value of $C$ is 0.75; the number $k$ is 0.25.\\
In Figure \ref{fonction_poly}(a), the function $f$ is intersected by the horizontal plane $z=0$. Figure \ref{fonction_poly}(b) represents the intersection in the variables frame.
\begin{figure}[!h]
\centering
\includegraphics[scale=0.4]{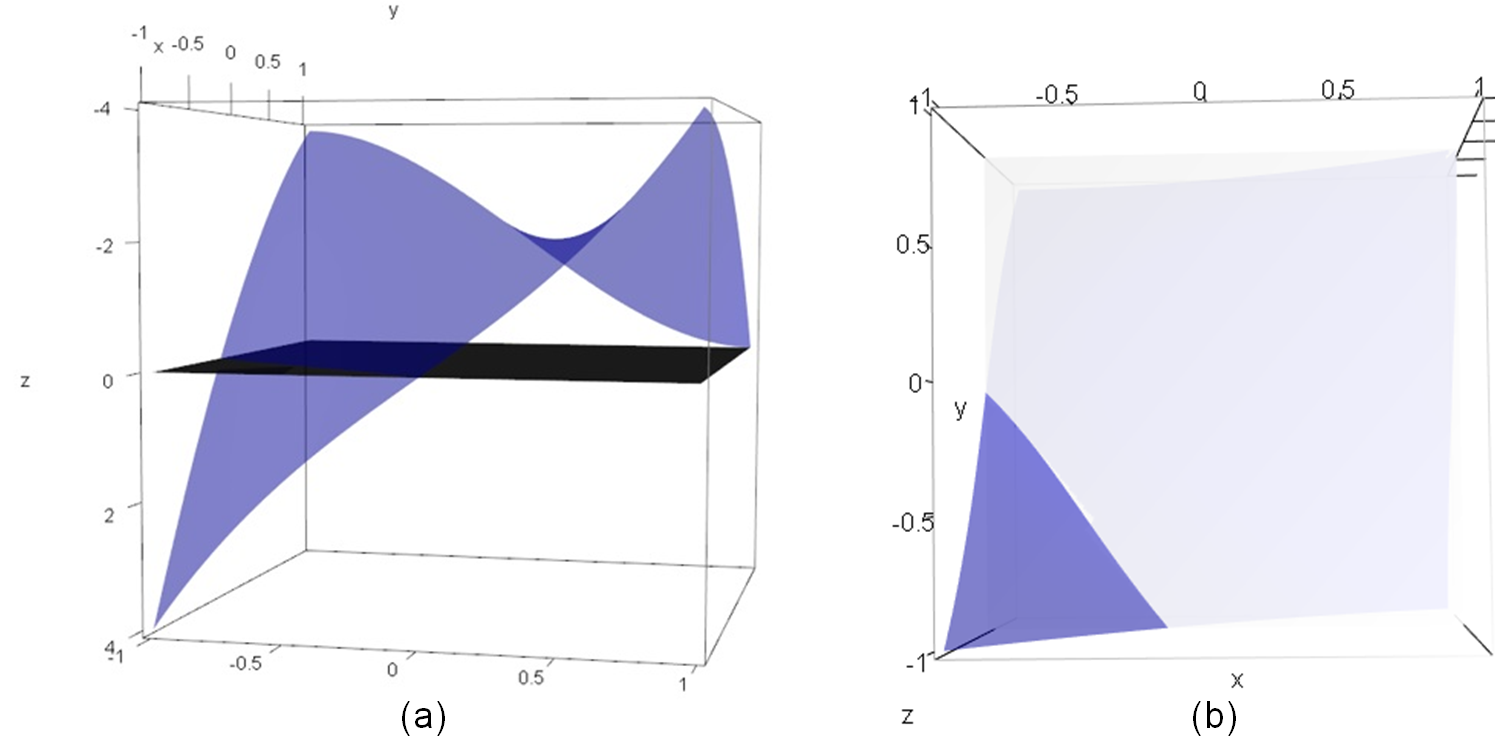}
\caption{Representations of the polynomial function}
\label{fonction_poly}
\end{figure}
$p$ is chosen as 1, 3 and 5. In Figures \ref{ex_p} (a), (b), (c), we see that the algorithm has produced some insight to elements in $S$ at the north-east region; however, the 1000 solutions have been obtained on the south-west component of $S$. Having asked for more solutions, we would have obtained the north-east component. Increasing $p$ to 3 or 5, the coefficient EC increases noticeably and the coverage of $S$ clearly increases.
\begin{figure}[!h]
\centering
\includegraphics[scale=0.5]{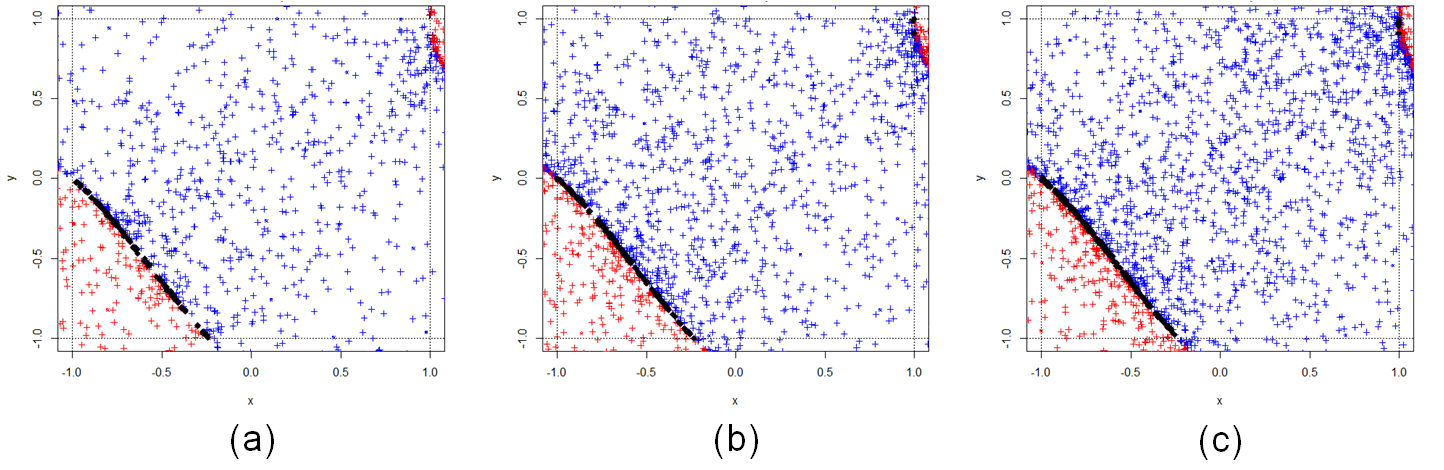}
\caption{Solving equation for the polynomial function using SAFIP for three values of $p$}
\label{ex_p}
\end{figure}
Results are gathered in Table \ref{table_ex4}.
\begin{table}
\begin{center}
\begin{tabular}{|c|c|c|c|c|c|c|c|c|}
\hline
\textbf{n} & \textbf{tol} & \textbf{N} & \textbf{C} & \textbf{k} & \textbf{p} & \textbf{Time} & \textbf{EC} & \textbf{Coverage} \tabularnewline
\hline
10 & 0.04 & 1000 & 0.75 & 0.25 & 1 & 2.12s & 15.94 & + \tabularnewline
\hline
10 & 0.04 & 1000 & 0.75 & 0.25 & 3 & 3.24s & 14.58 & + \tabularnewline
\hline
10 & 0.04 & 1000 & 0.75 & 0.25 & 5 & 4.96s & 17.01 & ++ \tabularnewline
\hline
\end{tabular}
\end{center}
\caption{Results for Example 4 with different values of $p$}
\label{table_ex4}
\end{table}
\end{example}

\hspace{1pt}
\paragraph{\textit{The tolerance factor $tol$}}
The strongest the tolerance (i. e. when $tol$ is small), the highest the number of evaluations of $f$, and the longest the runtime.
\begin{example}
Let $f$ be a bivariate function defined by
\begin{align*}
(x_1, x_2)\mapsto f(x_1,x_2)&=8\sin(7(x_1-0.9)^2)^2)+6\sin((14(x_1-0.9)^2)^2)+(x_1-0.9)^2\nonumber\\
&\hspace{0.1cm}+8\sin((7(x_2-0.9)^2)^2)+6\sin((14(x_2-0.9)^2)^2)\nonumber\\
&\hspace{0.1cm}+(x_2-0.9)^2-15
\end{align*} 
The aim is to find $N=1000$ pairs $(x_1, x_2)$ such that $\vert f(x_1, x_2)\vert\leq tol$ where $tol$ is the accuracy. All parameters but $tol$ are fixed. The number of initial points is 10; the value of $C$ is 0.75; the number $k$ is 0.08; the number $p$ of supplementary points at each step of the algorithm is 1.\\
In Figure \ref{fonction_trigo}(a), the function $f$ is intersected by the horizontal plane $z=0$. Figure \ref{fonction_trigo}(b) represents the intersection in the variables frame.
\begin{figure}[!h]
\centering
\includegraphics[scale=0.4]{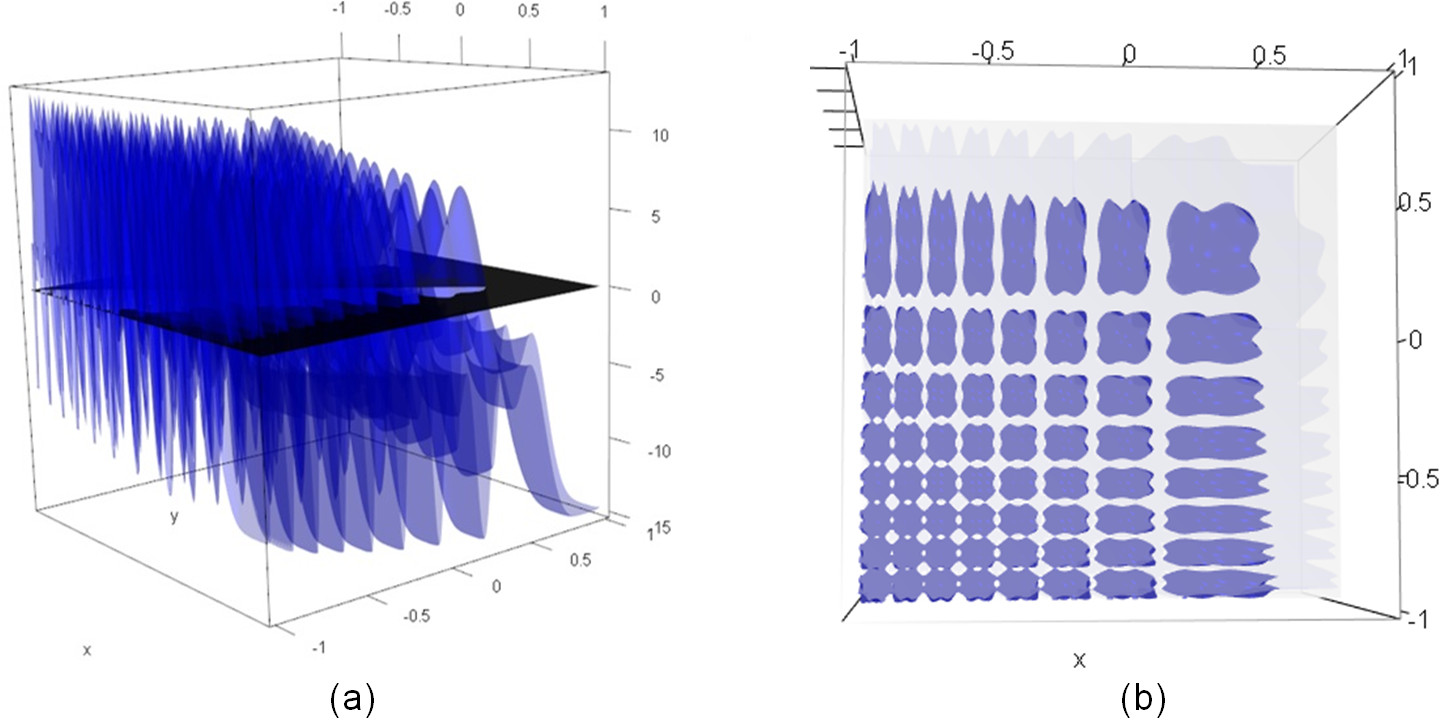}
\caption{Representations of the trigonometric function}
\label{fonction_trigo}
\end{figure}
The function oscillates between -15 and 15. In Figures \ref{ex_tol} (a), (b), (c), algorithm results are illustrated for three values of $tol$ : 0.15, 0.75 and 1.5 .
\begin{figure}[!h]
\centering
\includegraphics[scale=0.5]{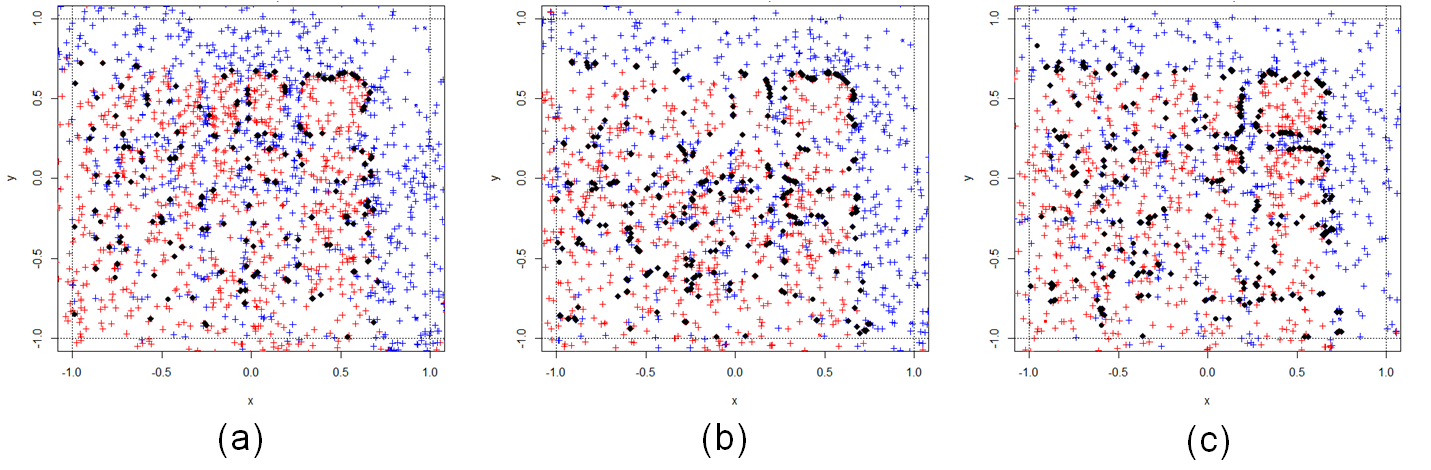}
\caption{Solving equation for the trigonometric function using SAFIP for three values of $tol$}
\label{ex_tol}
\end{figure}
When $tol$ varies from 0.15 to 1.5, the coefficient EC gets divided by 2. Results are gathered in Table \ref{table_ex5}.
\begin{table}
\begin{center}
\begin{tabular}{|c|c|c|c|c|c|c|c|c|}
\hline
\textbf{n} & \textbf{tol} & \textbf{N} & \textbf{C} & \textbf{k} & \textbf{p} & \textbf{Time} & \textbf{EC} & \textbf{Coverage} \tabularnewline
\hline
10 & 0.15 & 1000 & 0.75 & 0.25 & 1 & 2.6s & 43.47 & - \tabularnewline
\hline
10 & 0.75 & 1000 & 0.75 & 0.25 & 1 & 1.68s & 32.2 & - \tabularnewline
\hline
10 & 1.5 & 1000 & 0.75 & 0.25 & 1 & 1.3s & 22.85 & - \tabularnewline
\hline
\end{tabular}
\end{center}
\caption{Results for Example 5 with different values of $p$}
\label{table_ex5}
\end{table}
\end{example}
Due to the complexity of the function and of the set $S$, coverage is mild whatever $tol$; it depends upon the required number of solutions only. 

\hspace{1pt}
\paragraph{\textit{The role of $N$, the required number of solutions }}
The same function as in Example 4 is used in order to focus on the role of the number of solutions. When we ask for 15000 points in $S$, then the runtime remains quite satisfactory; the EC coefficient is 76, due to a choice of $n=1000$. The coverage of $S$ is quite fair.
\begin{figure}[!h]
\centering
\includegraphics[scale=0.5]{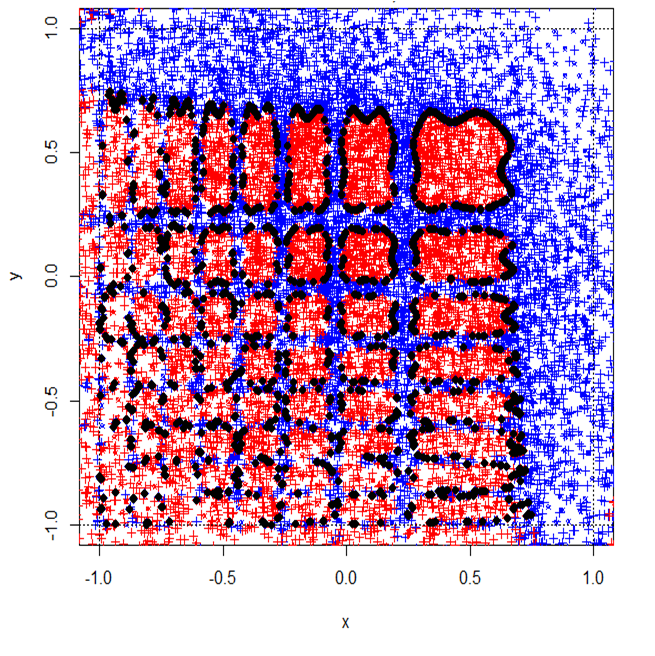}
\caption{Solving equation for the trigonometric function using SAFIP for a bigger number of required final points and a tolerance of 0.15}
\label{ex_tol2}
\end{figure}
Clearly the quality of the solutions improves with the required number of solutions. Not only do we get more solutions, but the coverage of $S$ improves noticeably. 
\begin{example}
Let $f$ be a bivariate function defined by
\[
(x_1, x_2)\mapsto f(x_1,x_2)=20+x_1^2-10\cos(2\pi x_1)+x_2^2-10\cos(2\pi x_2)-60
\]
The aim is to find $N$ pairs $(x_1, x_2)$ such that $\vert f(x_1, x_2)\vert\leq tol$ where $tol$ is the accuracy. All parameters but $N$ are fixed. The number of initial points is 10; $tol$ is fixed to 0.4; the value of $C$ is 0.75; the number $k$ is 0.025; the number $p$ of supplementary points at each step of the algorithm is 1.\\
In Figure \ref{fonction_rast}(a), the function $f$ is intersected by the horizontal plane $z=0$. Figure \ref{fonction_rast}(b) represents the intersection in the variables frame.
\begin{figure}[!h]
\centering
\includegraphics[scale=0.4]{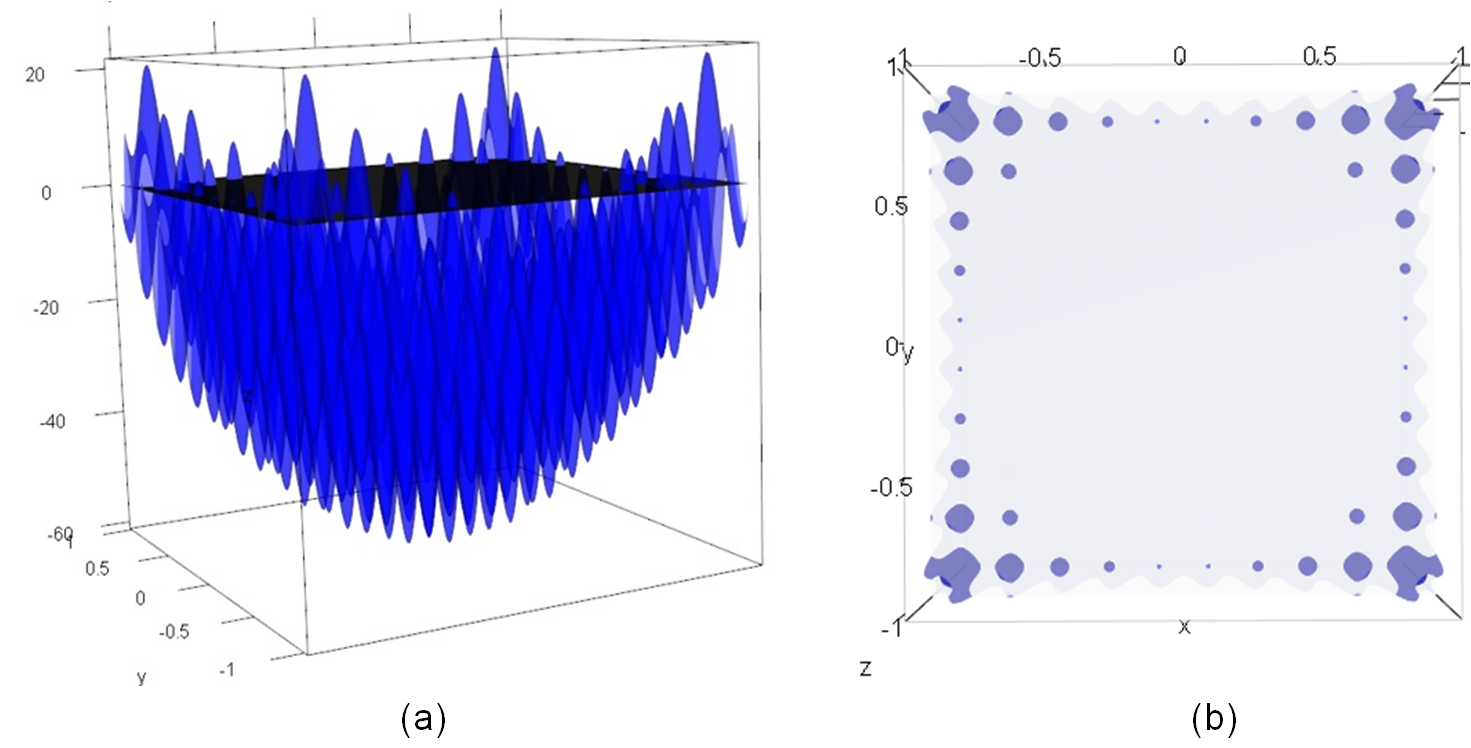}
\caption{Representations of the Rastrigin function}
\label{fonction_rast}
\end{figure}
In Figures \ref{ex_N_final} (a), (b), (c), algorithm results are illustrated for three values of $N$ : 100, 1000 and 2000.
\begin{figure}[!h]
\centering
\includegraphics[scale=0.5]{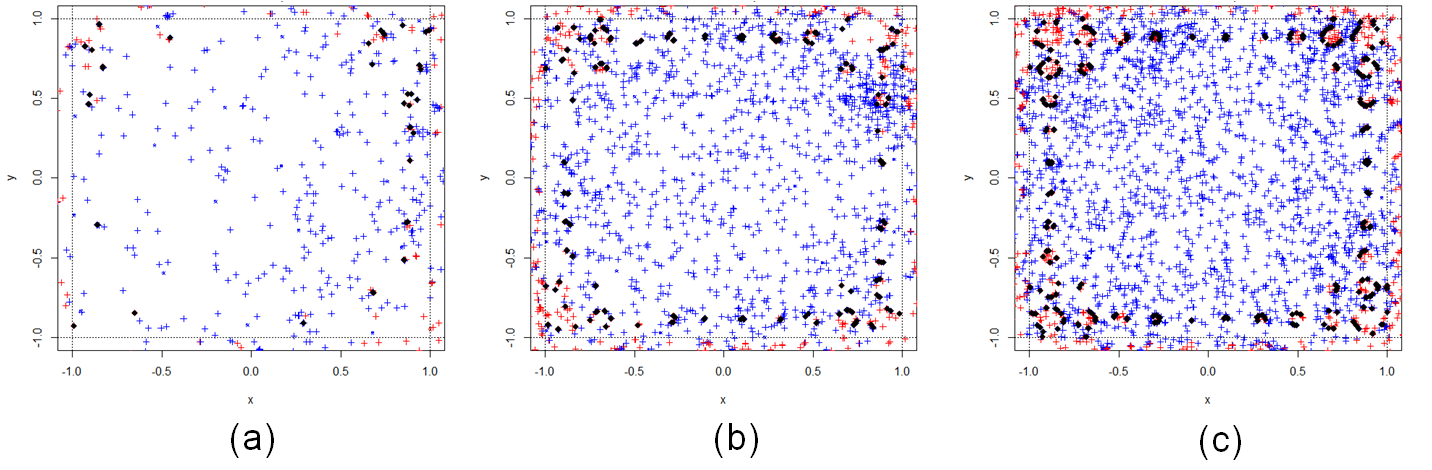}
\caption{Solving equation for the Rastrigin function using SAFIP for three values of $N$}
\label{ex_N_final}
\end{figure}
When $N$ is small, the important feature of the result is that $S$ is covered equally. So no cluster of solutions seems to appear; this is important for exploratory analysis. Results are gathered in Table \ref{table_ex6}.
\begin{table}
\begin{center}
\begin{tabular}{|c|c|c|c|c|c|c|c|c|}
\hline
\textbf{n} & \textbf{tol} & \textbf{N} & \textbf{C} & \textbf{k} & \textbf{p} & \textbf{Time} & \textbf{EC} & \textbf{Coverage} \tabularnewline
\hline
10 & 0.4 & 1000 & 0.75 & 0.025 & 1 & 0.48s & 55.33 & - \tabularnewline
\hline
10 & 0.4 & 1000 & 0.75 & 0.025 & 1 & 3.96s & 60.64 & - \tabularnewline
\hline
10 & 0.4 & 1000 & 0.75 & 0.025 & 1 & 8.6s & 83.82 & - \tabularnewline
\hline
\end{tabular}
\end{center}
\caption{Results for Example 6 with different values of $N$}
\label{table_ex6}
\end{table}
\end{example}

\subsection{Increasing the dimension}\label{sec dim+}
We consider a collection of functions which mimick Example 1, increasing the dimension. The required number of solutions is kept as $N=500$ in all cases.\\
We firstly consider the case in dimension 3, namely we look at points situated in
\begin{equation}\label{EQ S}
S:=\lbrace(x_1,x_2,x_3) : x_1^2+x_2^2+x_3^2-0.5=0\rbrace,
\end{equation}
with $-1\leq x_i\leq 1$ for $i=1, 2, 3$. The result appears in Figure \ref{sphere}.
\begin{figure}[!h]
\centering
\includegraphics[scale=0.4]{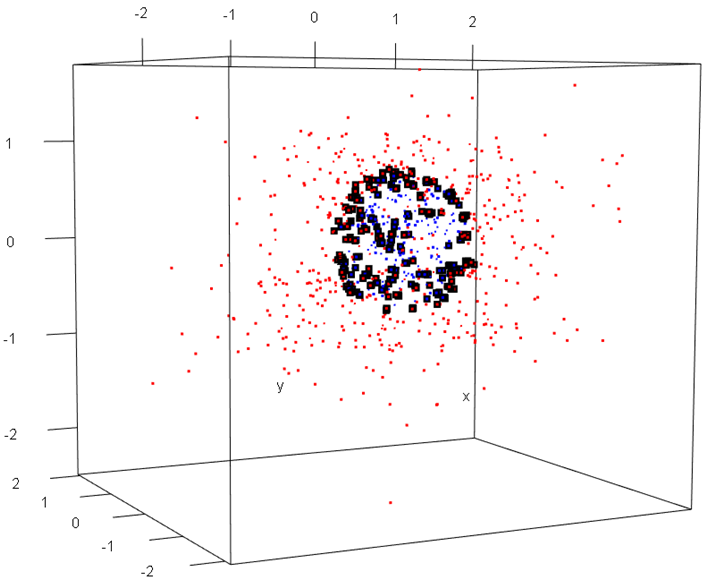}
\caption{Results for spheres in dimension 3}
\label{sphere}
\end{figure}
We also have considered the set 
\begin{equation}\label{EQ S2}
S:=\lbrace(x_1,x_2,x_3) : max(x_1,x_2,x_3)-0.5=0\rbrace;
\end{equation}
See Figure \ref{cube}.\\
\begin{figure}[!h]
\centering
\includegraphics[scale=0.4]{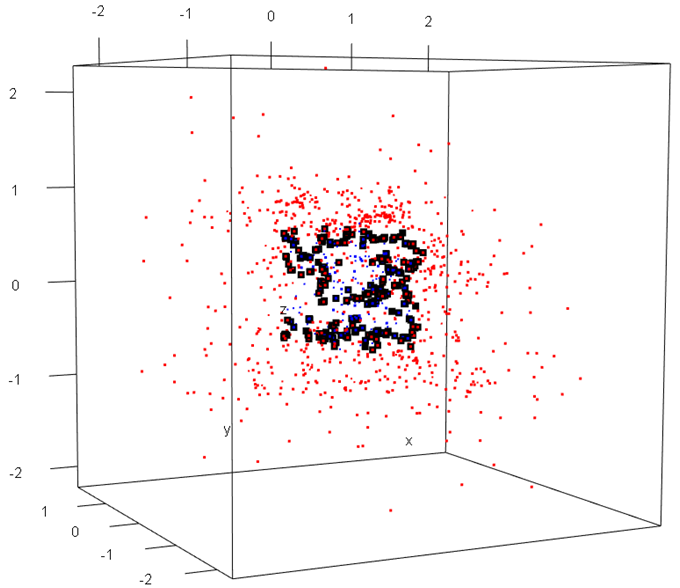}
\caption{Results for cubes in dimension 3}
\label{cube}
\end{figure}
Looking at similar examples as (\ref{EQ S}), we consider $d=4$ and $d=10$; the results comparing three dimensions are in Table \ref{tab_sphere}. 
\begin{table}
\begin{center}
\begin{tabular}{|c|c|c|c|c|c|c|c|c|}
\hline
\textbf{Dim} & \textbf{n} & \textbf{tol} & \textbf{N} & \textbf{C} & \textbf{k} & \textbf{p} & \textbf{Time} & \textbf{EC} \tabularnewline
\hline
2 & 5 & 0.1 & 500 & 0.75 & 1 & 1 & 0.22s & 4.81 \tabularnewline
\hline
3 & 25 & 0.1 & 500 & 0.75 & 1 & 1 & 4.72s & 6.64 \tabularnewline
\hline
4 & 75 & 0.1 & 500 & 0.75 & 1 & 1 & 0.4s & 9.7 \tabularnewline
\hline
10 & 1000 & 0.1 & 500 & 0.75 & 1 & 1 & 53s & 449 \tabularnewline
\hline
\end{tabular}
\end{center}
\caption{Results for spheres in different dimensions}
\label{tab_sphere}
\end{table}
The same is available for (\ref{EQ S2}) in Table \ref{tab_cube}.
\begin{table}
\begin{center}
\begin{tabular}{|c|c|c|c|c|c|c|c|c|}
\hline
\textbf{Dim} & \textbf{n} & \textbf{tol} & \textbf{N} & \textbf{C} & \textbf{k} & \textbf{p} & \textbf{Time} & \textbf{EC} \tabularnewline
\hline
2 & 5 & 0.1 & 500 & 0.75 & 1 & 1 & 0.16s & 4 \tabularnewline
\hline
3 & 25 & 0.1 & 500 & 0.75 & 1 & 1 & 4.2s & 5.04 \tabularnewline
\hline
4 & 75 & 0.1 & 500 & 0.75 & 1 & 1 & 0.72s & 8 \tabularnewline
\hline
10 & 1000 & 0.1 & 500 & 0.75 & 1 & 1 & 51s & 614 \tabularnewline
\hline
\end{tabular}
\end{center}
\caption{Results for cubes in different dimensions}
\label{tab_cube}
\end{table}
The number of initializing points has been chosen accordingly: $n=75$ for $d=4$, and $n=1000$ for $d=10$; a coherent choice for $n$ would have been $n=5^9$ for $d=10$, an impracticable choice.\\
Obviously the indicator EC increases with $n$. However, choosing $n=5^9$ and $N=500$, the value of EC exceeds 2000, which proves that $n$ should be kept low, growing slowly with respect to $d$. 

\section{Simultaneous inverse problems}
\label{sec_system}

\subsection{Algorithm}
Let $f$ and $g$ denote two functions defined on $D$; each of these functions $f$ and $g$ is assumes to satisfy hypothesis (\ref{EQ decr f}) together with conditions (\ref{cond1}) and (\ref{cond3}). We will make use of constants $C$, $k$, $n$ and $p$ defined in Section \ref{sec algo}; these constants will play a similar role in the present on $f$ and $g$. The number of common solutions to the system 
\begin{equation}
\left\{
  \begin{array}{rcr}
    f(x) & = & 0 \\
    g(x) & = & 0 \\
  \end{array}
\right.
\end{equation}
is denoted $N$.\\
Also the present section considers simultaneous inverse problems pertaining to two functions; quantization to a given number of functions is straightforward.\\
The algorithm is as follows with similar notation as in Section \ref{sec algo}, it holds
\begin{equation}
\left\{
  \begin{array}{rcr}
    f(x) & = & 0 \\
    g(x) & = & 0 
  \end{array}
\right.
\Leftrightarrow
\left\{
  \begin{array}{rcr}
    f(x)+\frac{x}{2k}+\frac{x}{2k} & = & \frac{x}{2} \\
    g(x)+\frac{x}{2k}+\frac{x}{2k} & = & \frac{x}{2} 
  \end{array}
\right.
\end{equation}
which yields to define
\begin{equation}
z_{i+1}=z_i+\frac{z_{i-1}-z_i}{2}+k\max(\vert f(z_i)\vert, \vert g(z_i)\vert).
\end{equation}
Inequality (\ref{EQ ineg dist}) is substituted by
\begin{equation}
R_{i+1}\leq\frac{R_i}{2}+kmax(\vert f(z_i)\vert, \vert g(z_i)\vert).
\end{equation}
Similarly as in (\ref{EQ rec}), the choice of $z_{i+1}$ follows the rule
\begin{equation}\label{EQ_rec_multi}
z_{i+1}=z_i+u_i
\end{equation}
where $u_i$ is drawn randomly on $\mathcal{B}(\underline{0}, \frac{R_i}{2}+kmax(\vert f(z_i)\vert,\vert g(x_i)\vert)$.\\
With those changes, denoting $S=\lbrace x:f(x)=0, g(x)=0\rbrace$, it holds
\begin{theorem}
Any sequence $(z_i)$ defined as above converges a. s. with limit in $S$.
\end{theorem}
and
\begin{theorem}
For any $x\in S$ and $\epsilon_0>0$ such that (\ref{cond1}) and (\ref{cond3}) hold simultaneously for $f$ and $g$, and when $z_0\in\mathcal{B}(x,\epsilon_0)$, thus the sequence $(z_n)$ is infinite and converges to $x$.
\end{theorem}

\subsection{Examples}
Due to  (\ref{EQ_rec_multi}), the point $z_{i+1}$ is randomly chosen in a ball $B$ centerd at $z_i$ when both $\vert f(z_i)\vert$ and $\vert g(z_i)\vert$ share a common measural order of magnitude. The best case is when $B$ has a moderate radius; it is therefore useful to normalize $f$ and $g$ on $D$; this preliminary procedure obviously does not modify the set $S$.\\
We present three examples of simultaneous inversion, based on the functions presented on Section \ref{sec algo}. In all examples the parameters are $n=20$, $p=1$, $tol=0.01$, $C=0.75$, $k=1$. $N$ equals 10 in Example \ref{ex_multi1}, it equals 100 in Example \ref{ex_multi2} and Example \ref{ex_multi3}.

\begin{example}[A regular case]\label{ex_multi1}
We choose $f$ as in Example \ref{ex_n_init} and $g(x)=f(x-a), a=(0.2, -0.2)$. Therefore $f(x)=0$ is as in Example \ref{ex_n_init} and $g(x)=0$ is a circle with same radius and center $a$.\\
Figures \ref{ex_multi_1_1}(a) and (b) show the graphs of $f$ and $g$ together with the intersection of the plane $z=0$. 
\begin{figure}[!h]
\centering
\includegraphics[scale=0.5]{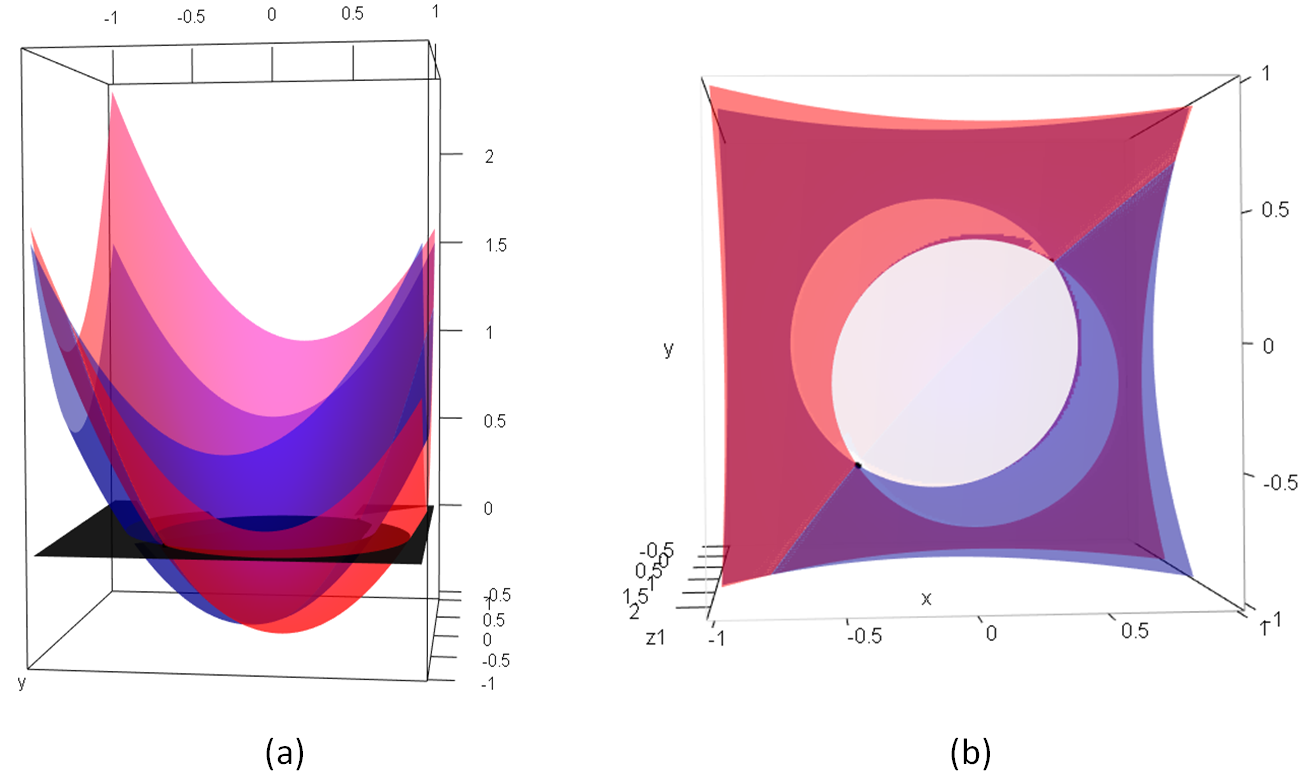}
\caption{Representations of $f$, $g$ and $S$}
\label{ex_multi_1_1}
\end{figure}
The set $S$ consists in the two points shown in Figure \ref{ex_multi_1_1}(b). Those points are indeed well estimated by the present algorithm, as seen in Figure \ref{ex_multi_1_2}.
\begin{figure}[!h]
\centering
\includegraphics[scale=0.5]{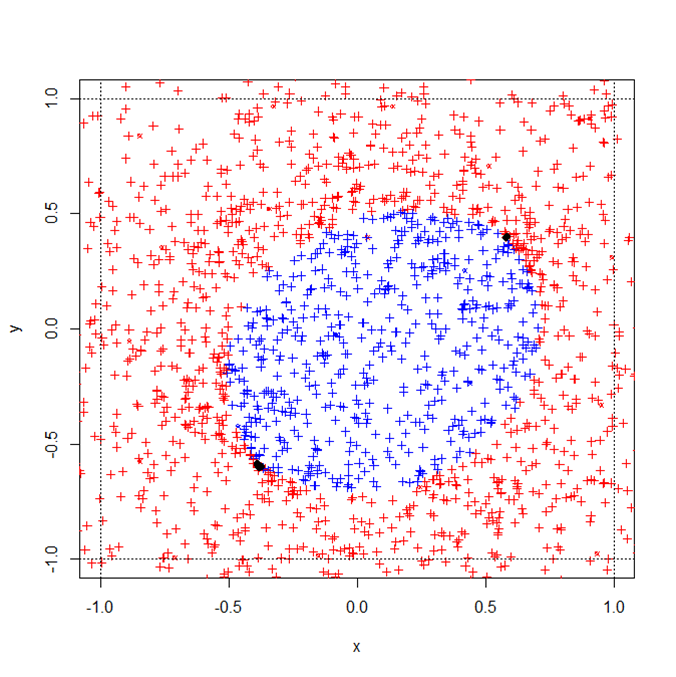}
\caption{Solutions obtained with SAFIP algorithm}
\label{ex_multi_1_2}
\end{figure}
The runtime is 0.62s and the efficiency coefficient is 516.
\end{example}

\begin{example}[Mixing a regular function and an irregular one]\label{ex_multi2}
We choose $f(x)$ as defined in Example \ref{ex_k}, a regular function, and $g(x)$ the Rastrigin function of Example \ref{ex_N_final}. The Figure \ref{ex_multi_2_1}(a) shows the two function, and Figure \ref{ex_multi_2_1}(b) provides the set $S$, which is defined as the intersection of the frontier points of the red domains (the solutions to $g(x)=0$) wt=ith the frontier points of the blue domains (the solutions to $g(x)=0$). There are 29 points in $S$.
\begin{figure}[!h]
\centering
\includegraphics[scale=0.5]{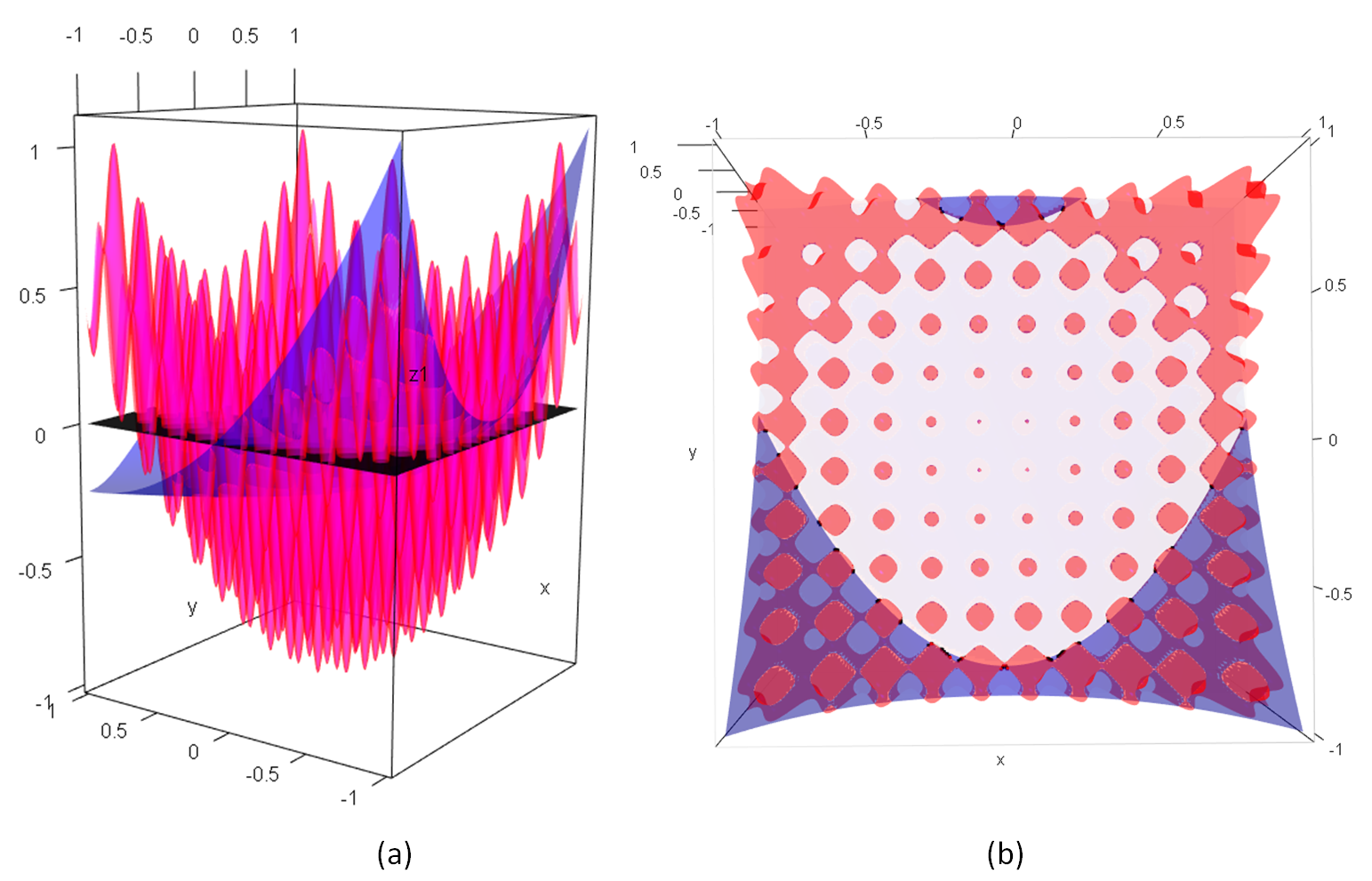}
\caption{Representations of $f$, $g$ and $S$}
\label{ex_multi_2_1}
\end{figure}
The algorithm provides solutions as shown in Figure \ref{ex_multi_2_2}, with runtime 14s and efficiency coefficient 375.
\begin{figure}[!h]
\centering
\includegraphics[scale=0.5]{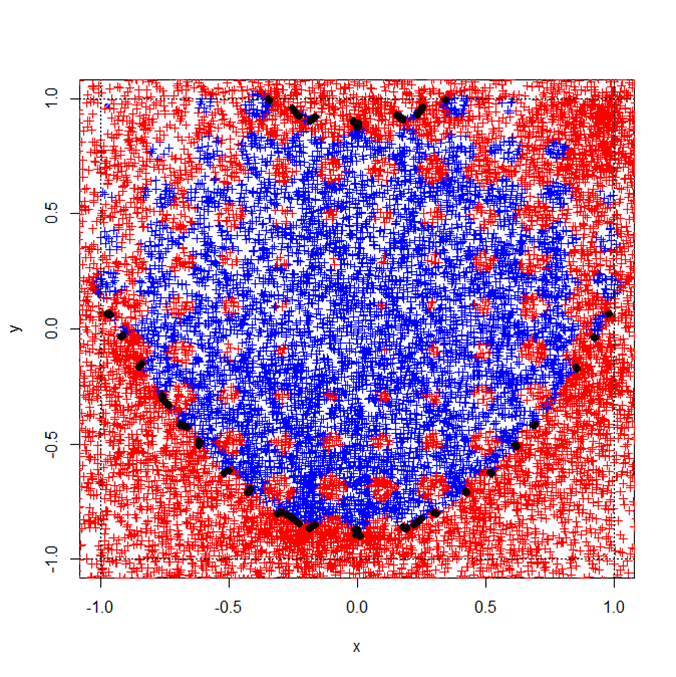}
\caption{Solutions obtained with SAFIP algorithm}
\label{ex_multi_2_2}
\end{figure}
Table \ref{table_ex_multi} provides results for different values of $C$, $k$ and $n$.\\
\begin{table}
\begin{center}
\begin{tabular}{|c|c|c|}
\hline
\textbf{C} & \textbf{EC} & \textbf{Temps} \tabularnewline
\hline
0.55 & 905 & 4.72s \tabularnewline
\hline
0.75 & 469 & 1.66s \tabularnewline
\hline
0.95 & 311 & 1.24s \tabularnewline
\hline
\textbf{k} & \textbf{EC} & \textbf{Temps} \tabularnewline
\hline
1 & 546 & 5.02s \tabularnewline
\hline
10 & 1963 & 8.8s \tabularnewline
\hline
50 & 6372 & 32.04s \tabularnewline
\hline
\textbf{n} & \textbf{EC} & \textbf{Temps} \tabularnewline
\hline
10 & 577 & 2.54s \tabularnewline
\hline
100 & 622 & 3.36s \tabularnewline
\hline
300 & 708 & 3.36s  \tabularnewline
\hline
\end{tabular}
\end{center}
\caption{Results for cubes in different dimensions}
\label{tab_cube}
\end{table}
As $C$ increases, EC decreases; as $k$ or $n$ increases, EC increases too.\\
A clear feature in Figure \ref{ex_multi_2_2} is that all the 29 points in $S$ are obtained a limiting points of SAFIP.
\end{example}

\begin{example}[A last example]\label{ex_multi3}
We choose $f(x)$ as in Example \ref{ex_n_init} and $g(x)$ the trigonometric function of Example \ref{ex_tol}. Figure \ref{ex_multi_3_1}(a) shows the functions $f$ and $g$; Figure \ref{ex_multi_3_1} (b) shows the intersection set $S$ which contains 33 points.
\begin{figure}[!h]
\centering
\includegraphics[scale=0.5]{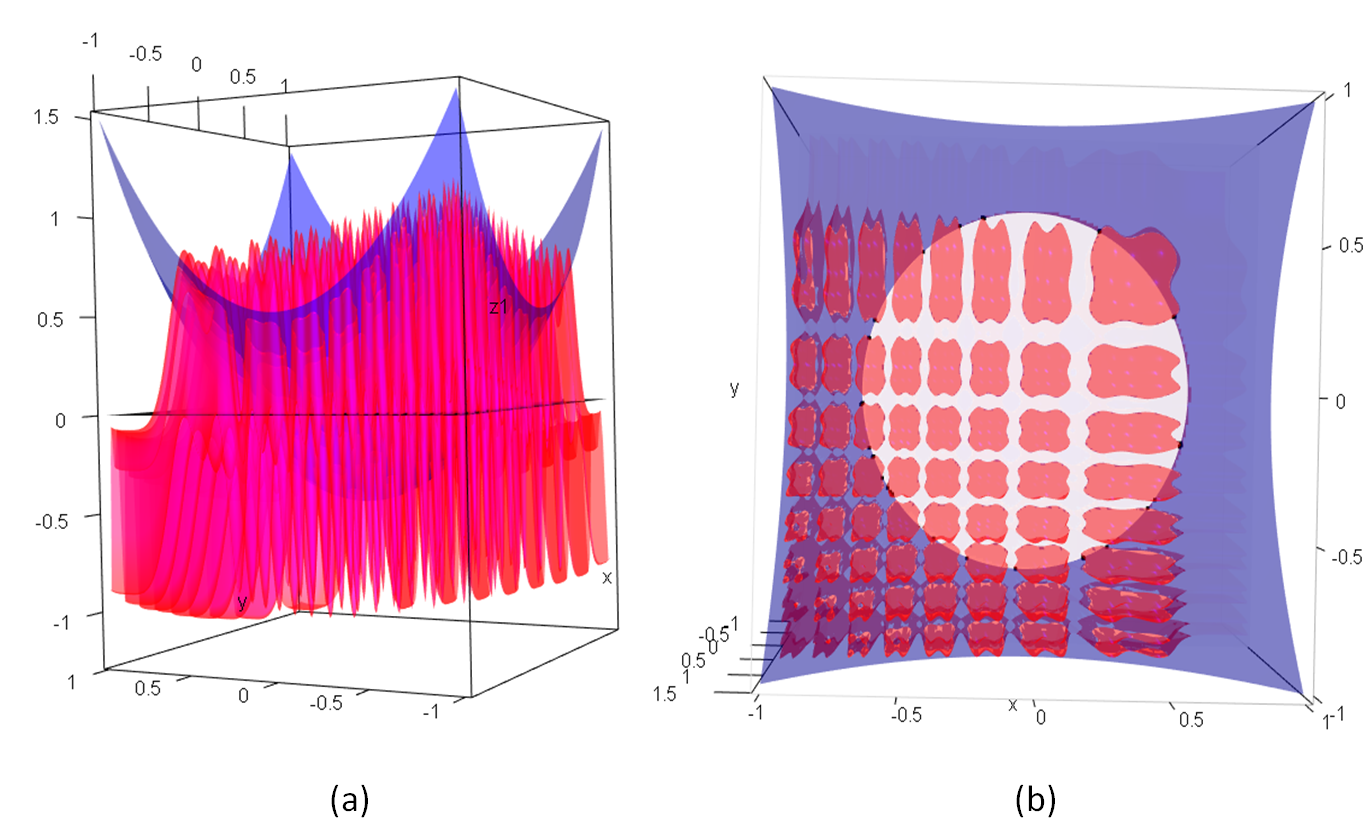}
\caption{Representations of $f$, $g$ and $S$}
\label{ex_multi_3_1}
\end{figure}
We asked for $N=100$ solutions; the set $S$ is not totally covered (we obtain 26 points in $S$ as it can be seen on Figure \ref{ex_multi_3_2}); a larger value of $N$ would provide all solutions
\begin{figure}[!h]
\centering
\includegraphics[scale=0.5]{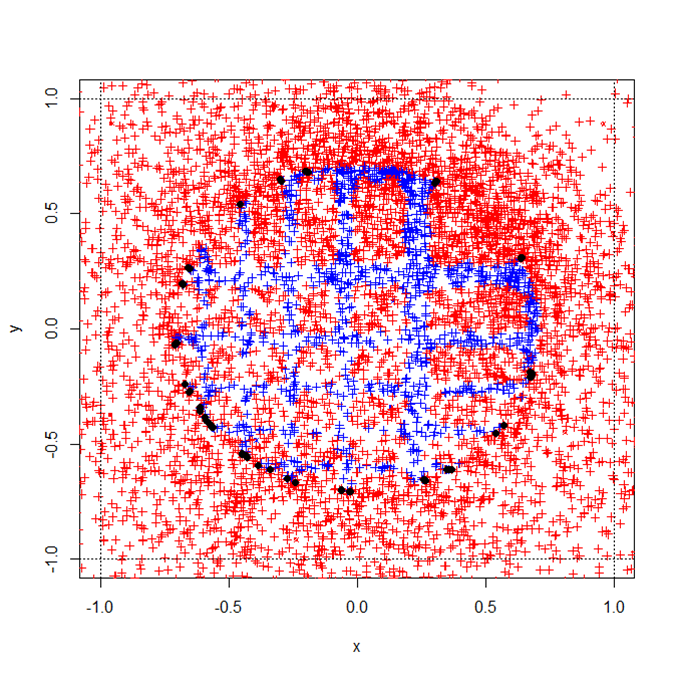}
\caption{Solutions obtained with SAFIP algorithm}
\label{ex_multi_3_2}
\end{figure}
The runtime is 4.1s and EC is 1102.
\end{example}

\section{Appendix}
\begin{proof}[Proof of Theorem \ref{thm1}]
\hspace{1pt}
\paragraph{\textit{Step 1}}
We prove that the sequence $(R_i)_{i\in\mathbb{N}}$ converges to 0 a. s.\\
Denote $a:=\vert f(z_0)\vert>0$. By (\ref{EQ decr f}),
\[
\vert f(z_i)\vert\leq aC^i,
\]
hence $R_{i+1}\leq\frac{R_i}{2}+akC^i$.\\
The sequence $(R_i)_{i\in\mathbb{N}}$ is now compared to the sequence $(x_i)_{i\in\mathbb{N}}$ defined by
\[
x_{i+1}=\frac{x_i}{2}+akC^i.
\]
It holds
\begin{align}
x_n&=\frac{x_0}{2^n}+\frac{ak}{2^{n-1}}+\frac{akC}{2^{n-2}}+\frac{akC^2}{%
2^{n-3}}+\ldots+\frac{akC^{n-2}}{2^1}+akC^{n-1}  \notag \\
&=\frac{x_0}{2^n}+akC^{n-1}\sum_{j=0}^{n-1}\left(\frac{1}{2C}\right)^j.
\label{EQ xn}
\end{align}
When $C>1/2$, it follows that $x_n$ given in (\ref{EQ xn}) tends to 0 as $n\rightarrow\infty$.\\
Since the generic term of $(R_n)_{n\in\mathbb{N}}$ satisfies
\begin{equation}  \label{EQ_R}
R_n\leq\frac{R_0}{2^n}+akC^{n-1}\sum_{j=0}^{n-1}\left(\frac{1}{2C}\right)^j,
\end{equation}
where the RHS is $x_n$, it follows that $R_n$ tends to 0 as $n\rightarrow\infty$.

\hspace{1pt}
\paragraph{\textit{Step 2}}
Assume at present that $(z_n)_{n\in\mathbb{N}}$ is an a. s. convergent sequence, and denote $l$ its limit. We prove that $l$ belongs to $S$. Indeed by (\ref{EQ recur}), writing $u_n=v_n(\frac{R_n}{2}+k\left\vert f(z_n)\right\vert)$ for $v_n$ uniformly distributed on $\mathcal{B}(0,1)$, the unit ball in $\mathbb{R}^d$. Going to the limit in (\ref{EQ recur}), $l=l+\underset{n\rightarrow\infty}{\lim}u_n$. It follows that $\underset{n\rightarrow\infty}{\lim}\frac{R_n}{2}+k\vert f(z_n)\vert=0$. Since $\underset{n\rightarrow\infty}{\lim}R_n=0$, it holds
\[
\underset{n\rightarrow\infty}{\lim}\vert f(z_n)\vert=0\hbox{ a. s}.
\]
By continuity of $f$, it follows that $\underset{n\rightarrow\infty}{\lim}\vert f(z_n)\vert=f(l)$ and then $f(l)=0$. We have proved that $l\in S$.\\
It remains to prove that $(z_n)_{n\in\mathbb{N}}$ converges, showing that it is a Cauchy sequence.\\
Let $(m,n)\in\mathbb{N}^2, m>n$. Then 
\begin{align*}
\sup_{m>n}\vert\vert
z_m-z_n\vert\vert&\leq\sup_{m>n}\sum_{j=n+1}^m\vert\vert
z_j-z_{j-1}\vert\vert \\
&\leq\sup_{m>n}\sum_{j=n+1}^m r_j.
\end{align*}
By (\ref{EQ_R}),
\begin{align*}
\sup_{m>n}\vert\vert z_m-z_n\vert\vert&\leq\sup_{m>n}\sum_{j=n+1}^m\left(%
\frac{r_0}{2^j}+akC^{j-1}\left(\frac{2C-\left(\frac{1}{2C}\right)^{j-1}}{2C-1%
}\right)\right) \\
&\leq\sup_{m>n}\left(\frac{r_0\left(1-\left(\frac{1}{2}\right)^{m-n}\right)}{%
2^{n}}+\frac{2akC^{n+1}}{2C-1}\times\frac{1-C^{m-n}}{1-C}-\frac{ak}{(2C-1)2^n%
}\times\frac{1-\left(\frac{1}{2}\right)^{m-n}}{\frac{1}{2}}\right),
\end{align*}
with $0<2C-1<1$. Since $m>n$ and $C<1$
\begin{equation*}
\sup_{m>n}\vert\vert z_m-z_n\vert\vert\leq\frac{r_0}{2^{n+1}}+\frac{2akC^{n+1}}{(2C-1)(1-C)}-\frac{ak}{(2C-1)2^{n-1}}
\end{equation*}
and therefore
\begin{equation}
\lim_{n\rightarrow\infty}\sup_{m>n}\vert\vert z_m-z_n\vert\vert=0,
\end{equation}
which proves the claim.

\end{proof}
\begin{proof}[Proof of Theorem \ref{thm2}]
By (\ref{EQ_E0}), we have $E_0=\lbrace z : k_1\vert f(z_0)\vert\leq\vert\vert z-z_0\vert\vert\leq \frac{R_0}{2}+k\vert f(z_0)\vert\lbrace$, with $\epsilon_0=\vert\vert x-z_0\vert\vert$. We have to prove that $E_0\cap A_1\not=\emptyset$.\\
By (\ref{EQ A1}) and since $E_0\subset B$, this is equivalent to prove that $\mathcal{B}(x,\epsilon_0)\cap E_0\not=\emptyset$. By the definition of $E_0$ which is an annulus centred on $z_0$ with a minimal radius of $2\epsilon_0$ and since $z_0\in\partial\mathcal{B}(x,\epsilon_0)$ according to the definition of $\epsilon_0$, $\mathcal{B}(x,\epsilon_0)\cap E_0\not=\emptyset$ and so $E_0\cap E_1\not=\emptyset$.\\
Let $z_1\in A_1\cap E_0$. we prove that $z_1$ satisfies (\ref{EQ decr f}).\\
By condition \ref{cond3}, it follows
\[
\vert f(z_0)\vert-\vert f(z_1)\vert\geq mk_1\vert f(z_0)\vert,
\]
since $z_1\in E_0$. This is equivalent to
\[
\vert f(z_1)\vert\leq(1-mk_1)\vert f(z_0)\vert
\]
With an arbitrary $k_1$ close to 0 such that $0<mk_1<\frac{1}{2}$. Getting $C=(1-mk_1)\in[\frac{1}{2},1]$, we have $\vert
f(z_1)\vert\leq C\vert f(z_0)\vert$ for $z_1\in E_0$. Thus $z_1\in A_1\cap\lbrace z_1, \vert f(z_1)\vert\leq C\vert f(z_0)\vert\rbrace$ and $z_0$ can have an offspring.\\
Iterating the above argument we can construct a sequence of balls $\mathcal{B}(x, \epsilon_i)$ with lower bounded and decreasing sequence of radius. Thus this sequence converges to some limit. By Theorem \ref{thm1}, $\underset{i\rightarrow\infty}{\lim} zi=x^*\in S$.\\
We show that $x^*=x$ by contradiction.\\
If $x^*\not=x$, thus there exists $i\in\mathbb{N}$ such that $x\not\in\mathcal{B}(x^*,\vert\vert x^*-z_i\vert\vert)$. But $z$ is simulated around $x$ with decreasing radius to 0. Hence is the contradiction. Thus $x^*=x$ and we have proved Theorem \ref{thm2}. 
\end{proof}

\bibliographystyle{plain}
\bibliography{COMET_chaine_biblio}

\end{document}